\documentclass[runningheads,a4paper, 11pt]{llncs}
\usepackage{amssymb}
\usepackage{graphicx}
\usepackage{url}
\usepackage{enumerate}
\usepackage{multirow}
\usepackage{mathtools}

\usepackage{array}
\newcolumntype{C}[1]{>{\centering\let\newline\\\arraybackslash\hspace{0pt}}m{#1}}
\urldef{\mailsa}\path|{vijay.menon, kate.larson}@uwaterloo.ca|

\usepackage[left=1.2in,top=1.2in,right=1.2in,bottom=1.2in]{geometry}

\begin{document}

\title{Complexity of Manipulation in Elections with Top-truncated Ballots}
\titlerunning{Complexity of Manipulation in Elections with Top-truncated Ballots}

\author{Vijay Menon \and Kate Larson}

\institute{David R. Cheriton School of Computer Science, University of Waterloo, \\ Waterloo, Ontario, Canada \\ \mailsa
}

\maketitle

\begin{abstract}
In the computational social choice literature, there has been great interest in understanding how computational complexity can act as a barrier against manipulation of elections. Much of this 
literature, however, makes the assumption that the voters or agents specify a complete preference ordering over the set of candidates. There are many multiagent systems applications, and even 
real-world elections, where this assumption is not warranted, and this in turn raises the question ``How hard is it to manipulate elections if the agents reveal only partial preference orderings?". 
It is this question we try to address in this paper. In particular, we look at the weighted manipulation problem -- both constructive and destructive manipulation -- when the voters are allowed to 
specify any top-truncated ordering over the set of candidates. We provide general results for all scoring rules, for elimination versions of all scoring rules, for the plurality with runoff rule, for 
a family of election systems known as Copeland$^{\alpha}$, and for the maximin protocol. Finally, we also look at the impact on complexity of manipulation when there is uncertainty about the 
non-manipulators' votes.  
\end{abstract}

\section{Introduction}
Preference aggregation is an important problem in multiagent settings as there are many scenarios where a group of agents has to make a common decision. The process of arriving at this decision, in 
turn, has to accommodate the needs and preferences of all the participating agents. A natural, and commonly used, mechanism to achieve this is voting, where all the agents specify their preferences 
and a previously agreed-upon procedure -- called the election system or voting protocol -- is used to arrive at the decision. Although voting is a useful and widely used mechanism, it is not without 
its problems. One particular issue that arises is manipulation or strategic voting by the agents who, by misreporting their true preferences, may attempt to sway the outcome of the election in their 
favor. Although every reasonable voting system is known to be manipulable by the Gibbard-Satterthwaite theorem, a goal of a recent body of literature -- which started with Bartholdi et. al's paper 
\cite{bartholdi} -- has been to understand if and when computational complexity can be used as a barrier against strategic voting (see \cite{falis3}, \cite{falis2} for surveys).

A common assumption in much of the research in computational social choice is that the agents fully specify their preferences by providing a complete preference ordering over all the candidates or 
alternatives. However, there are many practical situations where the agents may not be able to determine a complete ranking over all the candidates or even if they can specify a complete ranking, 
the voting rule used may not insist that one be provided. Thus, it is important to understand the repercussions of having ``partial votes'' in general. While a ``partial vote'' can refer to any 
partial ordering over the set of candidates, in this paper we focus on one kind of ``partial vote'' namely, top-truncated votes. Top-truncated votes are natural in many settings where an agent 
is certain about its most preferred candidates but is indifferent among the remaining ones or is unsure about them.

There has been some work that has looked at election problems when preferences are only partially specified. Among them, the work by Konczak and Lang introduced the possible and necessary winners 
problem \cite{konczak}, and Xia and Conitzer extended this further to study the possible and necessary winners problem for many different voting rules when the number of candidates are unbounded and 
the elections unweighted \cite{xia}. Additionally, Lu and Boutilier have looked at the multi-winner problem when only partial preferences are provided \cite{lu}. There are two other papers that are 
closely related to our work. First is the work by Baumeister et al. which discusses planning various kinds of campaigns in settings where the ballots can be truncated at the top, bottom or 
both \cite{baumeister}. In this work they introduced the extension-bribery problem, a special case of which is closely related to the manipulation problem with top-truncated ballots that we 
consider here. The second related work and the one which is the main motivation behind our work is that of Narodytska and Walsh where they provide an analysis of constructive manipulation (for both 
weighted and unweighted voters) for three particular voting protocols: Borda, STV, and the Copeland rule \cite{nina}. In addition to the above mentioned works, we are also aware of a very recent 
paper\footnote{This paper was obtained through personal communication.} by Fitzsimmons and Hemaspaandra which looks into how the complexity of bribery, control, and manipulation is affected when ties 
are allowed \cite{fitzsimmons}. We note that except for one theorem (Theorem \ref{copeland-alpha}), none of our other results overlap with theirs as all their results are derived using only one the 
following protocols: Borda, plurality, $t$-approval, and Copeland$^\alpha$. 

In this paper, we look at broader classes of voting rules and we study both constructive and destructive manipulation in weighted elections. In doing so, we provide general results for the complexity 
of manipulation for all scoring rules, for elimination versions of all scoring rules, for the plurality with runoff rule, for a family of election systems known as Copeland$^{\alpha}$, and for the 
maximin protocol. Additionally, we also look at the impact on complexity of manipulation with top-truncated ballots when there is uncertainty about the non-manipulators' votes, and, to the best of 
our knowledge, are the first to study the same. Because of space constraints, many of our proofs have been moved to the appendix.

\section{Preliminaries}
We model an election as a pair $E = (C, V)$, where $C = \{c_1, \cdots, c_m\}$ is the set of candidates and $V = \{v_1, \cdots, v_n\}$ is the set of voters. Each voter $v_i$ has a 
preference order $O_i$ on $C$. $O_i$ is said to be a complete order (or a complete vote) when it is antisymmetric, transitive, and a total ordering on $C$. In this paper we also consider 
top-truncated orders (or simply top orders), meaning that $O_i$ can be a linear order over any non-empty subset of $C$ and where all the unranked candidates are tied and are assumed to be ranked 
below the ranked candidates. For example, consider an election scenario with $C = \{c_1, c_2, c_3\}$. A voter $v_i$ who prefers $c_2$ the best and dislikes $c_1$ the most has a complete ordering $O_i$ 
which is represented by $(c_2 \succ c_3 \succ c_1)$ or sometimes simply $(c_2, c_3, c_1)$, while another voter $v_j$ who likes $c_3$ but has no opinion on $c_1$ and $c_2$ has a top-truncated ordering 
$O_j$ given by $(c_3)$. A preference profile is a vector $P = \langle O_1, \cdots, O_n\rangle$ of individual preferences. Since this paper considers weighted manipulation, additionally every voter 
$v_i$ has a non-negative integer weight $w_i$ associated with them. 

\subsection{Voting Protocols}
A voting protocol is a function defined from the set of all preference profiles to the set of winners, where the winner set can be any subset over the set of candidates $C$. The following are the 
commonly-studied voting rules that we consider in this paper. For each of them, we first define them on complete orderings and then talk about how the evaluations are done when there are top orders. 
\begin{enumerate}
 \item \textbf{Positional scoring rules:} A positional scoring rule is defined by a scoring vector $\alpha = \langle \alpha_1, \allowbreak \cdots, \alpha_m \rangle$, where $\alpha_1 \geq \cdots \geq 
\alpha_m$. For each voter $v$, a candidate receives $\alpha_i$ points if it is ranked in the $i$th position by $v$. In a scoring rule, the candidate with highest total score $s_i$ is the winner. Some 
examples of scoring rules are the \textit{plurality rule} with $\alpha = \langle 1, 0, \cdots, 0\rangle$, the \textit{Borda rule} with $\alpha = \langle m-1, m-2, \cdots, 0\rangle$, and the 
\textit{veto rule} with $\alpha = \langle 1, \cdots, 1, 0\rangle$.

To deal with top-truncated ballots where a voter ranks only $k$ out of the $m$ candidates ($k < m$), we consider the following three schemes that were used by Narodytska and Walsh in their 
preliminary work on manipulation with top-truncated preferences \cite{nina} (see \cite{emerson} to see how the following schemes fare for the Borda rule). As we will show in the 
results, the choice of evaluation scheme does have an impact on the hardness of manipulation.

\begin{enumerate}
 \item \textbf{Round up:} A candidate ranked in the $i$th position ($i\leq k)$ receives a score of $\alpha_i$, while all the unranked candidates receive a score of $\alpha_m$. For example, consider 
an election with $C = \{c_1, c_2, c_3, c_4\}$. Let a voter $v$ give a preference ordering $(c_3, c_1)$. In this case, $c_3$ receives a score of $\alpha_1$, $c_1$ receives a score of $\alpha_2$, and 
both $c_2$ and $c_4$ receive $\alpha_4$. For any positional scoring rule $X$, we denote this by $X_{\uparrow}$.  

 \item \textbf{Round down:} A candidate ranked in the $i$th position ($i\leq k)$ receives a score of $\alpha_{m-(k-i)-1}$, while all the unranked candidates receive a score of $\alpha_m$. For the 
example above, in this case $c_3$ receives a score of $\alpha_2$, $c_1$ receives a score of $\alpha_3$, and both $c_2$ and $c_4$ receive a score of $\alpha_4$. For any positional scoring rule 
$X$, we denote this by $X_{\downarrow}$. 

 \item \textbf{Average score:} A candidate ranked in the $i$th position ($i\leq k)$ receives a score of $\alpha_{i}$, while all the unranked candidates receive a score of $\frac{\textstyle\sum_{k < j 
\leq m}\alpha_j}{m - k}$. For the example above, in this case $c_3$ receives a score of $\alpha_1$, $c_1$ receives a score of $\alpha_2$, and both $c_2$ and $c_4$ receive a score of 
$\frac{\alpha_3 + \alpha_4}{2}$. For any positional scoring rule $X$, we denote this by $X_{av}$. 
\end{enumerate}

  \item \textbf{Scoring elimination rules:} Let $X$ be any scoring rule. Given a complete ordering, eliminate($X$) is the rule that successively eliminates the candidate placed in the last place by 
$X$. Once a candidate is eliminated, the rule is then repeated with the reduced set of candidates until there is a single candidate left. Some examples of scoring elimination rules are the 
\textit{Single Transferable Vote} (STV) which is basically eliminate(plurality) and the \textit{Baldwin's rule} which is eliminate(Borda).

In scoring elimination rules, we deal with top-truncated votes by using a method given by Narodytska and Walsh which is analogous to rounding up for scoring rules \cite{nina}. Here, we consider a 
vote to be valid only until at least one of the candidates listed in it is remaining in the election. In other words, we simply ignore a vote once all the candidates listed in it are eliminated.

\item \textbf{Plurality with runoff:} The plurality with runoff rule proceeds in two steps. In the first step, all the candidates except the top two with the most number of first votes are 
eliminated. This is followed by a transfer of votes to the second round where the winner is determined using the majority rule. 

Top-truncated votes here are dealt in the same way as they are done for scoring elimination rules.

  \item \textbf{Copeland$^\alpha$:} Faliszewski et al. \cite{falis} introduced a family of election systems known as Copeland$^{\alpha}$ by introducing a parameter $\alpha$ ($\alpha \in 
\mathbb{Q}$, $0\leq \alpha \leq 1$) that essentially describes the value of a tie. In Copeland$^{\alpha}$, for each pair of candidates, the candidate preferred by the majority receives one point and 
the other one receives a 0. In case of a tie, both receive $\alpha$ points. The winner in a Copeland$^\alpha$ election is one with the highest score. Usually when papers use the Copeland rule they 
essentially mean Copeland$^{0.5}$ which was the original rule proposed by Copeland. 

For Copeland$^\alpha$, and the maximin rule below, we can deal with top-truncated votes by just sticking to the definition which assumes that all the unranked candidates are tied and are ranked below 
the ranked candidates.

\item \textbf{Maximin:} Let $N_P(c_i, c_j)$ denote the number of voters who prefer $c_i$ over $c_j$ in the preference profile $P$. Then the maximin score of $c_i$ is $s_i = \min_{j \neq 
i} \allowbreak N_{P}(c_i, c_j)$. The winner in the maximin rule is the one with the highest score.

\end{enumerate}

\subsection{Manipulation}
In this paper, we consider two kinds of manipulation: constructive manipulation and destructive manipulation. Broadly, the goal in the former is to make a preferred candidate win, while in the 
latter it is to ensure that a certain disliked candidate does not win. More formally, we consider constructive weighted coalitional manipulation and destructive weighted coalitional 
manipulation which was first studied by Contizer et al. \cite{conitzer} and are described below.

\begin{definition}[CWCM]
In Constructive Weighted Coalitional Manipulation (CWCM), given a set of weighted votes $S$ (votes of the non-manipulators), the weights for a set of votes $T$ (manipulators' votes), and a preferred 
candidate $p$, we are asked if there exists a way to cast the votes in $T$ so that $p$ wins the election. 
\end{definition}

In this paper, unless otherwise specified, all the results are based on the non-unique winner model (where the objective is to make $p$ \textit{a} winner) which we use as our standard model. 

\begin{definition}[DWCM]
In Destructive Weighted Coalitional Manipulation (DWCM), given a set of weighted votes $S$ (votes of the non-manipulators), the weights for a set of votes $T$ (manipulators' votes), and a disliked 
candidate $h$, we are asked if there exists a way to cast the votes in $T$ so that $h$ does not win the election. 
\end{definition}

\subsection{Computational Complexity}
In most of the proofs for NP-hardness in this paper, we use reductions from either the well-known NP-complete problem Partition or from a variant of the subset sum problem which we 
call Fixed-Difference Subset Sum. 

\begin{definition}[{Partition}]
 Given a set of non-negative integers $S = \{k_i\}_{1 \leq i \leq t}$ summing to $2K$, we are asked if there exists a subset $S_1$ of $S$ which sums to $K$. 
\end{definition}

\begin{definition}[{Fixed-Difference Subset Sum}]
 Given a set of non-negative integers $S = \{k_i\}_{1 \leq i \leq t}$ summing to $2K$, we are asked if there exists two disjoint subsets $S_1$, $S_2$ of $S$ such that $\sum S_1 - \sum S_2 = K$, where 
$\sum S_i$ denotes the sum of all the elements in the set $S_i$. 
\end{definition}
The NP-completeness of Fixed-Difference Subset Sum can be shown by a reduction from Partition. 
\begin{theorem} \label{FDSS}
 Fixed-Difference Subset Sum is NP-complete.
\end{theorem}

\section{Constructive Manipulation} \label{constructivemanipulation}
In this section we look at the complexity of constructive manipulation when top-truncated ballots are allowed. We begin by looking at scoring rules and we completely characterize the complexity of 
manipulation for all 3-candidate scoring rules when using each of the evaluation schemes defined before. Note that in the complete votes case, all scoring rules except plurality are known to be 
NP-complete for $m\geq 3$ candidates \cite{conitzer,hemaspaandra}.

\subsection{Scoring Rules}

\begin{theorem} \label{SC-RU}
 For any positional scoring rule $X$, computing if a coalition of manipulators can manipulate $X_{\uparrow}$ with weighted top-truncated votes takes polynomial time (for any number of candidates). 
\end{theorem}

\begin{proof}
 The manipulators can simply check if all of them voting for $p$ alone will make it a winner. If not, they cannot make $p$ a winner. \qed
\end{proof}

\begin{theorem}
 For the plurality and veto protocol, computing if a coalition of manipulators can manipulate plurality$_{\downarrow}$ or veto$_{\downarrow}$ with weighted top-truncated votes takes 
polynomial time (for any number of candidates). 
\end{theorem}

\begin{proof}
 For the veto protocol, the manipulators can simply check if all of them voting for $p$ alone will make it a winner. If not, they cannot make $p$ a winner.
 
 In case of plurality, they can check if all of them placing $p$ at the top and all the other candidates in arbitrary order can make $p$ a winner. \qed
\end{proof}

\begin{theorem} \label{rdThm}
 For any 3-candidate positional scoring protocol $X$ that is not isomorphic to plurality or veto, CWCM with top-truncated votes in $X_{\downarrow}$ is NP-complete.
\end{theorem}

\begin{theorem}
Computing if a coalition of manipulators can manipulate plurality~$_{av}$ with weighted top-truncated votes takes polynomial time (for any number of candidates). 
\end{theorem}
\begin{proof}
 The manipulators can simply check if all of them voting for $p$ alone will make it a winner. If not, they cannot make $p$ a winner. \qed
\end{proof}

\begin{theorem} \label{avThm}
 For any 3-candidate positional scoring protocol $X$ that is not isomorphic to plurality, CWCM with top-truncated votes in $X_{av}$ is NP-complete.
\end{theorem}

\paragraph{}
Although both Theorem \ref{rdThm} and Theorem \ref{avThm} can be proved by constructing much simpler instances (by using the fact that even the non-manipulators can cast top-truncated votes), the 
reason we use these instances will be clear in Section \ref{uncertaintyNM} when we look at the impact on manipulation when there is uncertainty about the non-manipulators' votes. 

\subsection{Scoring Elimination Rules}
We now consider scoring elimination rules and first look at how top-truncated voting affects the complexity of manipulation in eliminate(veto). Following that we prove a general result for all 
scoring elimination rules. 

\begin{theorem} \label{vetoThm}
 For eliminate(veto), in the unique winner model, any manipulation that can be achieved by casting top-truncated votes can be achieved if only complete votes were allowed. 
\end{theorem}


Since top-truncated voting does not encourage more strategic voting, it follows that for bounded number of candidates we can use the result by Coleman and Teague who showed that CWCM for 
eliminate(veto) can is in $P$ when the votes are completely specified \cite{coleman}.  

\begin{corollary} \label{vetoCorr}
 In the unique winner model, computing if a coalition of manipulators can manipulate eliminate(veto) with weighted top-truncated votes takes polynomial time for bounded number of candidates. 
\end{corollary}

We note that in the non-unique winner model, CWCM for 3-candidate eliminate(veto) can be shown to be NP-complete in both the complete and top-truncated voting scenarios.  

Next, we consider elimination versions of scoring rules in general and show that CWCM with top-truncated votes for $m \geq 3$ candidates is NP-complete for elimination version of any scoring 
rule that is not isomorphic to veto. For this, we first show that top-truncated voting does not change the complexity of Anti-WCM for any scoring rule. Subsequently, we use this result and an 
identical reduction as in \cite{coleman} to prove our main result.

\begin{definition}[{Anti-WCM}]
Given a set $S$ of weighted votes, the weights for a set of votes $T$, and a disliked candidate $d$, we are asked if there exists a way to cast the votes in $T$ so that it results in $d$ receiving 
the lowest score. 
\end{definition}

\begin{theorem} \label{antiwcm}
Top-truncated voting does not change the worst-case complexity of Anti-WCM for any scoring protocol.
\end{theorem}

As a result of the above theorem, we have the following corollary which says that for any scoring rule not isomorphic to veto, Anti-WCM with top-truncated votes is NP-complete. Note that the 
corollary here is based on the result of Coleman and Teague who proved that Anti-WCM is NP-complete for all scoring rules not isomorphic to veto \cite{coleman}.

\begin{corollary} \label{corollaryAntiWCM}
For any scoring rule with $\alpha = \langle \alpha_1, \cdots, \allowbreak \alpha_m \rangle$, Anti-WCM with top-truncated votes is in P if $\alpha_1 = \allowbreak \cdots = \alpha_{m-1}$ and is 
NP-complete otherwise.
\end{corollary}

\begin{theorem} \label{eliminateNPC}
 For any scoring rule $X$ that is not isomorphic to veto, CWCM with top-truncated votes in eliminate($X$) is NP-complete.
\end{theorem}

\begin{proof}
 From Corollary \ref{corollaryAntiWCM} we know that Anti-WCM with top-truncated votes is NP-complete for any scoring rule that is not isomorphic to veto. Therefore we can use the exact same technique 
as in \cite[Theorem 13]{coleman} where they use a reduction from an arbitrary instance of Anti-WCM to prove the above result for the case of complete votes. \qed
\end{proof}

Since the plurality with runoff rule is the same as STV when there are only three candidates, we have the following corollary.

\begin{corollary}
 For the 3-candidate plurality with runoff rule, CWCM with top-truncated votes is NP-complete. 
\end{corollary}

\subsection{Copeland$^{\alpha}$}
Narodytska and Walsh showed that CWCM with top-truncated votes in the Copeland rule (Copeland$^{0.5}$) is NP-complete for four candidates \cite{nina}. Additionally, they also conjectured that the 
result holds when the number of candidates is three. Here we prove that conjecture, and also show that our hardness result holds for all rational $\alpha \in [0, 1)$. We note that the following 
result has also been independently obtained by Fitzsimmons and Hemaspaandra in a very recent paper (obtained through personal communication) \cite{fitzsimmons}. 

\begin{theorem} \label{copeland-alpha}
 Let $\alpha$ be a rational number with $0 \leq \alpha < 1$. For Copeland$^{\alpha}$, CWCM with top-truncated votes is NP-complete for three candidates. 
\end{theorem}

\subsection{Maximin}
Although we have seen instances like in scoring rules with the round-up evaluation scheme (Theorem \ref{SC-RU}) were top-truncated voting decreases the complexity of manipulation (as compared to the 
complete votes case), they aren't really conclusive in the sense that one could question the choice of the evaluation scheme that in turn caused the result. Clearly, the round-up evaluation scheme 
isn't a good choice as what we're essentially doing by employing the same is to encourage the manipulators to treat it just as a plurality-type rule. Therefore, one question that could arise is: ``Is 
there a voting system for which no matter how the top-truncated votes are dealt with the complexity of manipulation with top-truncated ballots decreases?''. We answer this question below in the 
affirmative. We see that as long as all the unranked candidates are considered tied and are assumed to be ranked below the ranked candidates (which is the natural definition of a top-truncated 
vote), the complexity of manipulation with top-truncated ballots is in $P$ for the maximin rule. Note that CWCM for the maximin rule is known to be NP-complete for four candidates when we consider 
only complete votes \cite[Theorem 8]{conitzer}\footnotemark.

\footnotetext{Although Conitzer et al. uses the unique winner model \cite{conitzer}, it can be verified that the result holds for non-unique winner model as well.}

\begin{theorem} \label{maxminThm}
 Computing if a coalition of manipulators can manipulate the maximin protocol with weighted top-truncated votes takes polynomial time (for any number of candidates). 
\end{theorem}
To prove this, we show that any constructive manipulation achieved for $p$ can be achieved if all manipulators just vote $(p)$. The detailed proof can be found in the appendix.

\section{Destructive Manipulation}
In this section we look at the complexity of destructive manipulation when top-truncated ballots are allowed. We begin by looking at a broad class of rules for which top-truncated voting has no 
impact on strategic voting. This class consists of all voting rules where the candidates are assigned numerical scores based on the votes and are monotone, meaning that if a voter $v_i$ changes his 
or 
her vote (from $\succ$ to $\succ'$) in such a way that $\{b : a \succ b \} \subseteq  \{b : a \succ' b\}$, then a's score will not decrease (or informally, more support for a candidate will not 
decrease it's score). Note that although elimination versions of scoring rules like STV and the Baldwin's rule are based on numerical scores, they are not monotone and hence aren't part of the class 
of voting rules we consider in the theorem below. 

\begin{theorem} \label{dwcmThm1}
 For any voting rule that is monotone and is based on numerical scores, any destructive manipulation that can be achieved by casting top-truncated votes can be achieved if only complete votes were 
allowed.
\end{theorem}


Since top-truncated voting has no impact on destructive manipulation in rules that are monotone and are based on numerical scores, it follows that for bounded number of candidates we can use the 
result by Conitzer et al. who showed that DWCM was in $P$ for all of them when only complete votes are allowed \cite{conitzer}.

\begin{corollary}
 DWCM with top-truncated votes is in $P$ for all scoring rules, for the maximin rule, and for Copeland$^{\alpha}$.
\end{corollary}

Next, we consider the elimination versions of all 3-candidate scoring rules and we show how for all scoring rules that are not isomorphic to veto DWCM is NP-complete. We prove this by using a 
reduction from an arbitrary instance of Anti-WCM.

\begin{theorem} \label{eliminateNPCdest}
 For any 3-candidate positional scoring rule $X$ that is not isomorphic to veto, DWCM in eliminate(X) with top-truncated votes is NP-complete.
\end{theorem}

Again, since the plurality with runoff rule is equivalent to STV when there are only three candidates, we have the following corollary.
\begin{corollary}
 For the  3-candidate plurality with runoff rule, DWCM with top-truncated votes is NP-complete. 
\end{corollary}

\section{Impact on Complexity of Manipulation when there is Uncertainty about the Non-manipulators' Votes} \label{uncertaintyNM}
So far, we have looked at the complexity of manipulation with top-truncated ballots when the manipulators have complete information on the non-manipulators' votes. Although a useful setting to study 
given that it enables us to look at the hardness of manipulation without having to worry about the complexities that are introduced as part of the uncertainty model, the assumption that the 
manipulators will have complete information isn't always realistic. Therefore, we now look at how incomplete information about the non-manipulators impact the complexity of manipulation 
with top-truncated ballots. To the best of our knowledge, we are the first to look at complexity of manipulation with top-truncated ballots when there is uncertainty about the others' votes. 

To model the incomplete information setting, we consider the following two scenarios and we study each of them separately. Note that the second setting has been studied by Conitzer et al. for 
the case of complete votes \cite{conitzer}.
\begin{enumerate}
 \item What if only top-truncated preference orderings of the non-manipulators were visible to the manipulators?
 \item What if the manipulators have only probabilistic information on the votes of the non-manipulators?
\end{enumerate}

\subsection{When only top orders of the non-manipulators are visible} \label{TT-NM}
Before we look at the problem of manipulation, let us introduce two other problems: i) the problem of evaluating a candidate's winning probability when there's uncertainty about the votes and ii) the 
weighted version of the extension-bribery problem which in turn was introduced by Baumeister et al. \cite{baumeister}. As has been the case throughout this paper, we only consider elections in  which 
the voters are weighted. 

\begin{definition}[Evaluation under Top-truncated Uncertainty]
We are given a set $S$ which is a partially-revealed top-truncated set of votes of a certain set $S'$ (which in turn may contain complete or top-truncated ballots themselves), the weights of each of 
the voters, a candidate $p$, and a number $r \in [0, 1]$. We are asked if the probability of $p$ winning in the original election (where $S'$ is the set of votes cast) is greater than 
$r$.  
\end{definition}

\begin{definition}[Weighted Extension-bribery]
We are given a set $S$ of votes which are possibly top-truncated, the weights of each of the voters in $V$, a collection $\Delta = (\delta_1, \cdots, \delta_n)$ of extension-bribery cost functions, a 
preferred candidate $p$, and a budget $B$. We are asked if there exists an extension to the votes in $S$ with cost $\leq B$ such that $p$ is the winner.  
\end{definition}

Now, we show that the Evaluation under Top-truncated Uncertainty problem with $r = 0$ (henceforth also referred to as the Evaluation problem) is equivalent to a special case of the Weighted 
Extension-bribery problem namely, Weighted Extension-bribery with zero costs (i.e. when $\delta_i = 0, \: 1 \leq i \leq n $). 

\begin{theorem} \label{extbribEval}
 For a given voting protocol, Weighted Extension-bribery with zero costs is NP-hard if and only if the Evaluation under Top-truncated Uncertainty problem (with $r = 0$) is 
NP-hard.
\end{theorem}

Next, before we look at the complexity of Weighted Extension-bribery with zero costs (and hence of Evaluation) for all the voting rules considered in this paper, consider the following version 
(CWCM') of CWCM with top-truncated votes where the only difference is that here we make it necessary for the non-manipulators to always have complete ballots (i.e. in CWCM' only the manipulators can 
caste top-truncated votes). 

\begin{definition}[CWCM']
CWCM' with top-truncated votes is the same problem as CWCM with top-truncated votes, with the additional restriction that the non-manipulators always have complete preference orders.  
\end{definition}

The first thing to observe here is that all the NP-complete results for CWCM with top-truncated votes from Section \ref{constructivemanipulation} hold for CWCM' as well since a close look at the 
proofs for Theorem \ref{rdThm}, Theorem \ref{avThm}, Theorem \ref{eliminateNPC}, and Theorem \ref{copeland-alpha} reveal that in all the cases we showed reductions to instances which always had 
complete orders for the non-manipulators. Hence, we have the following theorem. 

\begin{theorem}
 CWCM' with top-truncated votes is NP-complete for 3-candidate $X^1_{\downarrow}$, 3-candidate $X^2_{av}$, eliminate($X^3$), 3-candidate plurality with runoff rule, and 3-candidate Copeland$^\alpha$ 
for $\alpha \in [0, 1)$, where $X^1$ represents all scoring rules except plurality and veto, $X^2$ represents all scoring rules except plurality, and $X^3$ represents all scoring rules except veto.  
\end{theorem}

The second observation to make regarding CWCM' with top-truncated votes is that it is a special case of Weighted Extension-bribery with zero costs, where some the ballots are complete (those 
of the non-manipulators) and some are empty (those of the manipulators). Therefore, any hardness result for CWCM' carries over for Weighted Extension-bribery with zero costs (and also for the 
Evaluation problem as a consequence of Theorem \ref{extbribEval}), and so we have the following theorem.

\begin{theorem} \label{extbrib-NPC}
 Weighted extension-bribery problem with zero costs (and hence even the Evaluation problem) is NP-complete for 3-candidate $X^1_{\downarrow}$, 3-candidate $X^2_{av}$, eliminate($X^3$), 3-candidate 
plurality with runoff, and 3-candidate Copeland$^\alpha$ for $\alpha \in [0, 1)$, where $X^1$ represents all scoring rules except plurality and veto, $X^2$ represents all scoring rules except 
plurality, and $X^3$ represents all scoring rules except veto. 
\end{theorem}

Additionally, we also have the following result which shows that Weighted Extension-bribery with zero costs is in P for certain voting rules considered in this paper.

\begin{theorem} \label{ext-eval}
  Weighted Extension-bribery  with zero costs (and hence even the Evaluation problem) is in $P$ for $X_{\uparrow}$, plurality$_{\downarrow}$, veto$_{\downarrow}$, plurality$_{av}$, and the 
maximin protocol, where $X$ is any positional scoring rule.
\end{theorem}

Next we show that if Evaluation is hard then constructive manipulation with even a single manipulator is hard. 

\begin{definition}[CWIM-TTU] 
In Constructive Weighted Individual Manipulation under Top-truncated Uncertainty (CWIM-TTU), we are given a set $S$ of partially-revealed top-truncated ballots of the non-manipulators (which in turn 
may contain complete or top-truncated ballots themselves), the weight of the manipulator, a preferred candidate $p$, and a number $r \in [0, 1]$. We are asked if the manipulator can cast his vote in 
such a way so as to ensure $p$ wins with a probability greater than $r$.  
\end{definition}

\begin{theorem} \label{CWIMhard}
 If Evaluation under Top-truncated Uncertainty is NP-hard for a given protocol with $k$ candidates, then CWIM-TTU with top-truncated votes is also NP-hard for it with $k$ candidates.  
\end{theorem}
\begin{proof}
 Construct an instance of CWIM-TTU from an Evaluation instance by just adding a manipulator of weight 0. \qed
\end{proof}
 
Combining Theorem \ref{extbrib-NPC} and Theorem \ref{CWIMhard}, we have the following theorem which says that for all the protocols considered in Section \ref{constructivemanipulation} for which 
CWCM with top-truncated ballots was hard, even individual manipulation with top-truncated votes is hard when there is uncertainty about the non-manipulators' votes. 

\begin{theorem}
CWIM-TTU with top-truncated votes is NP-complete for 3-candidate $X^1_{\downarrow}$, 3-candidate $X^2_{av}$, eliminate($X^3$), 3-candidate plurality with runoff rule, and 3-candidate Copeland$^\alpha$ 
for $\alpha \in [0, 1)$, where $X^1$ represents all scoring rules except plurality and veto, $X^2$ represents all scoring rules except plurality, and $X^3$ represents all scoring rules except veto. 
\end{theorem}

Finally, we conclude this section by showing that the CWIM-TTU with top-truncated votes is in $P$ for eliminate(veto) (it is easy to see that it is also in $P$ for 
all the rules mentioned in Theorem \ref{ext-eval}). We show this by first proving that Weighted Extension-bribery with zero costs (and hence Evaluation) is in $P$ for eliminate(veto). 

\begin{theorem} \label{elVeto}
In the unique winner model, Weighted Extension-bribery with zero costs (and Evaluation) is in $P$ for eliminate(veto) when the number of candidates is bounded. 
\end{theorem}
%

\begin{theorem} \label{eVCWIM}
 In the unique winner model, CWIM-TTU with top-truncated votes (with $r = 0$) is in $P$ for eliminate(veto) when the number of candidates is bounded. 
\end{theorem}

The above result also implies that, in the unique winner model, CWCM under Top-truncated Uncertainty and with top-truncated votes (with $r = 0$) is in $P$ for eliminate(veto) when the number of 
candidates is bounded, since in eliminate(veto) any manipulation that can be induced by an arbitrary set of the manipulators' vote can be induced if all the manipulators vote in the same way 
\cite[Lemma 12]{coleman}.

\subsection{When there is only probabilistic information on the non-manipulators' votes}
Manipulation under probabilistic uncertainty was introduced by Conitzer et al. \cite{conitzer}. In their work they introduced the Weighted Evaluation problem which, given a probability 
distribution on the votes, and a number $r \in [0,1]$, asks if the probability of a candidate winning is greater than $r$. Subsequently they also proved that if CWCM (with complete votes) for a voting 
protocol is hard then so is Weighted Evaluation \cite[Theorem 15]{conitzer}. Now, since we additionally allow top-truncated votes, we can state the following result which is almost equivalent to 
\cite[Theorem 15]{conitzer} with the only difference being that the ``set of all possible votes'' would now contain top-truncated votes as well and that the reduction here is from CWCM with 
top-truncated votes. In fact, most of the results in this section are only extensions to the corresponding result from Conitzer et al.'s  paper \cite{conitzer} that arise as a result of allowing 
top-truncated voting.

\begin{theorem}
 If CWCM with top-truncated votes is NP-hard for a given protocol with $k$ candidates, then Weighted Evaluation when top-truncated votes are allowed is also NP-hard for it (with $k$ candidates) even 
if $r=0$, the votes are drawn independently, and only the following types of distributions are allowed: (1) the vote's distribution is uniform over all possible votes, or (2) the vote's distribution 
puts all of the probability mass on a single vote.   
\end{theorem}
\begin{proof}
 We can proceed exactly as in \cite[Theorem 15]{conitzer} to prove this result. \qed
\end{proof}

\begin{corollary}
 Weighted Evaluation when top-truncated votes are allowed is NP-hard for 3-candidate $X^1_{\downarrow}$, 3-candidate $X^2_{av}$, eliminate($X^3$), 3-candidate plurality with runoff, and 
3-candidate Copeland$^\alpha$ for $\alpha \in [0, 1)$, when the votes are drawn independently, and the distributions allowed are: (1) uniform over all possible votes, or (2) the vote's distribution 
puts all of the probability mass on a single vote, where $X^1$ represents all scoring rules except plurality and veto, $X^2$ represents all scoring rules except plurality, and $X^3$ represents all 
scoring rules except veto.  
\end{corollary}

Next, we show a relation between the Weighted Evaluation and manipulation with a single manipulator as in \cite{conitzer}. Note that the difference between CWIM with Uncertainty (CWIM-U) that we 
define below and CWIM-TTU that we defined in the Section \ref{TT-NM} is in how the partial information about the non-manipulators' votes are specified. 

\begin{definition}[CWIM-U] 
In Constructive Weighted Individual Manipulation under Uncertainty (CWIM-U), given a distribution over all the non manipulators' votes, the weights of the non-manipulators, the weight of the 
manipulator, a preferred candidate $p$, and a number $r \in [0, 1]$, we are asked if the manipulator can cast his vote in such a way so as to ensure $p$ wins with a probability greater than $r$.  
\end{definition}

\begin{theorem}
 If Weighted Evaluation with top-truncated votes is NP-hard for a protocol with $k$ candidates and some restrictions on the distribution, then CWIM-U with top-truncated votes is also NP-hard for it 
with $k$ candidates and the same restrictions on the distribution.
\end{theorem}
\begin{proof}
  Construct an instance of CWIM-U from an arbitrary Weighted Evaluation instance by just adding a manipulator of weight 0. \qed
\end{proof}

Finally, we show that Weighted Evaluation when top-truncated votes are allowed can be hard even if CWCM with top-truncated votes is in $P$. For this we consider eliminate(veto) for which CWCM with 
top-truncated  votes was shown to be in $P$ for bounded number of candidates (see Corollary \ref{vetoCorr}). 

\begin{theorem} \label{eV-WE}
 In eliminate(veto), Weighted Evaluation when top-truncated votes are allowed is NP-hard even if $r = 0$, the votes are drawn independently, and the distribution over each vote has a positive 
probability for at most 2 of the votes. 
\end{theorem}

%


\section{Conclusion and Future Work}

Much of the earlier work in computational social choice has made the assumption that the voters specify complete preference orderings. However, there are many situations where the agents may not be 
willing or may simply not be able to provide this information. In this paper, we studied the problem of manipulation of weighted elections when the agents are allowed to specify top-truncated 
preferences and also looked at the impact on manipulation when there is uncertainty about the non-manipulators' votes. We devoted most of this article studying the first problem and in particular we 
provided general results for constructive and destructive manipulation in all scoring rules, elimination versions of all scoring rules, the plurality with runoff rule, a family of election systems 
known as Copeland$^{\alpha}$, and the maximin protocol. These results are summarized in Table \ref{results}. 

As was also noted by Narodytska and Walsh in their study of manipulation with top-truncated votes for Borda, STV, and Copeland$^{0.5}$ \cite{nina}, there are three broad trends that we can observe. 
First is the case where top-truncated voting has a strong impact on manipulation and it in turn results in a decrease in the worst-case complexity of manipulation as compared to the complete votes 
case. Examples of this are all the scoring rules when using the round-up evaluation scheme and the maximin rule. Second is the case where top-truncated voting has some impact on manipulation and in 
fact even causes more strategic voting, but yet the worst-case complexity of manipulation remains the same as compared to the complete votes case. Some examples of voting rules which fall into this 
category are the Copeland$^{\alpha}$, $X_{\downarrow}$ for any scoring rule $X$ that is not isomorphic to plurality or veto etc. Lastly, we also see that there are voting rules for which 
top-truncated voting has no impact whatsoever on strategic voting. For instance, top-truncated voting has no impact on STV and eliminate(Veto). 

\begin{table}[t]
  \centering  

 \begin{tabular}{| c |c|c|} 
    \cline{1-3} 
    \cline{1-3}    

    \multicolumn{1}{|c|}{\multirow{2}{*}{\textbf{  Voting Rule  }}} & \textbf{  CWCM  } & \textbf{  DWCM  }\\ 
    \multicolumn{1}{|c|}{} & (\#cand) & (\#cand) \\ \cline{1-3}   

    \multicolumn{1}{|c|}{\textbf{$X_{\uparrow}$}} & \textbf{P} & \multicolumn{1}{c|}{\multirow{6}{*}{P}} \\ \cline{1-2}
    \multicolumn{1}{|c|}{$\text{Plurality}_{\downarrow}$} & P  & \multicolumn{1}{c|}{}\\
    \multicolumn{1}{|c|}{\textbf{$\text{Veto}_{\downarrow}$}} & \textbf{P} & \multicolumn{1}{c|}{}\\   
    \multicolumn{1}{|c|}{\textit{$X^1_{\downarrow}$}} & \textit{NP-c (3)} & \multicolumn{1}{c|}{} \\ \cline{1-2}
    \multicolumn{1}{|c|}{$\text{Plurality}_{\text{av}}$} & P  & \multicolumn{1}{c|}{} \\
    \multicolumn{1}{|c|}{\textit{$X^2_{\text{av}}$}} & \textit{NP-c (3)} & \multicolumn{1}{c|}{} \\ \hline   

   \multicolumn{1}{|c|}{eliminate(Veto)} & P & \multicolumn{1}{c|}{P} \\
   \multicolumn{1}{|c|}{eliminate($X^3$)} & NP-c & \multicolumn{1}{c|}{NP-c (3)} \\ \hline

   \multicolumn{1}{|C{2.3cm}|}{Plurality with runoff} & NP-c (3) & \multicolumn{1}{c|}{NP-c (3)} \\ \cline{1-3}    
   \multicolumn{1}{|C{2.3cm}|}{\textit{Copeland$^{\alpha}$, $\alpha \in [0, 1)$}} & \textit{NP-c (3)} & \multicolumn{1}{c|}{P} \\ \cline{1-3}  
   \multicolumn{1}{|c|}{\textbf{Maximin}} & \textbf{P} & \multicolumn{1}{c|}{P} \\
   \cline{1-3}   
   \end{tabular}      
 
  \begin{tabular}{ll}     
     $X:$ All scoring rules & $X^1:$ All scoring rules except plurality and veto \\
     $X^2:$ All scoring rules except plurality & $X^3:$ All scoring rules except veto \medskip
    \end{tabular}   
           
           \caption{Complexity of CWCM and DWCM with Top-truncated Votes. The entries in bold indicate that there is a change in complexity of CWCM as compared to the case of complete votes, the 
non-highlighted ones indicate that there is no change to the worst-case complexity as compared to the case of complete votes, and the italicized entries indicate those rules for which there is more 
opportunity for manipulation but for which the worst-case complexity is still the same as compared to the case of complete votes.}  
 \label{results} 
 
\end{table}

The remainder of this article was devoted to exploring the second avenue i.e. when there is uncertainty about the non-manipulators' votes. Here we discussed two possible ways in which the uncertainty 
can be modeled and we also showed that in both cases, under uncertainty, even individual manipulation was hard when constructive coalitional manipulation was hard. To the best of our knowledge, we 
are the first to study the impact on manipulation with top-truncated ballots when there is uncertainty about the non-manipulators' votes. 

There are many possible avenues for future work. Foremost would be to look at other forms of ``partial votes'' besides the top-truncated ones alone that we consider here. Although Fitzsimmons and 
Hemaspaandra have started some work in this direction \cite{fitzsimmons}, they only look at specific protocols. We feel that it would be interesting to look at how other voting rules behave and 
also see if there is any scope to arrive at some general results. Another interesting problem could be to look at restricted domains of preferences like single-peaked preferences and see the impact 
of top-truncated voting on them. Finally, despite proving hardness of manipulation for several voting protocols, a possible criticism could be that the results are in the worst-case and that we 
use NP-hardness -- which in many cases does not necessarily reflect the actual difficulty in practice -- as the complexity measure. Therefore, another interesting research direction would be to look 
at the average-case complexity of manipulation with top-truncated votes (or ``partial'' votes in general), just like it has been done for the case of complete votes (for e.g. \cite{friedgut}, 
\cite{procaccia2}, \cite{procaccia3}).

\bibliographystyle{llncsbib}
\bibliography{paper}

\begin{thebibliography}{10}
\providecommand{\url}[1]{\texttt{#1}}
\providecommand{\urlprefix}{URL }

\bibitem{bartholdi}
Bartholdi~III, J.J., Tovey, C.A., Trick, M.A.: The computational difficulty of
  manipulating an election. Social Choice and Welfare  6(3),  227--241 (1989)

\bibitem{baumeister}
Baumeister, D., Faliszewski, P., Lang, J., Rothe, J.: Campaigns for lazy
  voters: Truncated ballots. In: Proceedings of the 11th International
  Conference on Autonomous Agents and Multiagent Systems-Volume 2. pp.
  577--584. International Foundation for Autonomous Agents and Multiagent
  Systems (2012)

\bibitem{coleman}
Coleman, T., Teague, V.: On the complexity of manipulating elections. In:
  Proceedings of the thirteenth Australasian symposium on Theory of
  computing-Volume 65. pp. 25--33. Australian Computer Society, Inc. (2007)

\bibitem{conitzer}
Conitzer, V., Sandholm, T., Lang, J.: When are elections with few candidates
  hard to manipulate? Journal of the ACM (JACM)  54(3), ~14 (2007)

\bibitem{emerson}
Emerson, P.: The original borda count and partial voting. Social Choice and
  Welfare  40(2),  353--358 (2013)

\bibitem{falis3}
Faliszewski, P., Hemaspaandra, E., Hemaspaandra, L.A.: Using complexity to
  protect elections. Communications of the ACM  53(11),  74--82 (2010)

\bibitem{falis}
Faliszewski, P., Hemaspaandra, E., Schnoor, H.: Copeland voting: Ties matter.
  In: Proceedings of the 7th international joint conference on Autonomous
  agents and multiagent systems-Volume 2. pp. 983--990. International
  Foundation for Autonomous Agents and Multiagent Systems (2008)

\bibitem{falis2}
Faliszewski, P., Procaccia, A.D.: Ai's war on manipulation: Are we winning? AI
  Magazine  31(4),  53--64 (2010)

\bibitem{fitzsimmons}
Fitzsimmons, Z., Hemmaspandra, E.: Complexity of manipulative actions when
  voting with ties. In: To appear at Algorithmic Decision Theory (2015)

\bibitem{friedgut}
Friedgut, E., Kalai, G., Nisan, N.: Elections can be manipulated often. Hebrew
  University, Center for the study of Rationality (2008)

\bibitem{hemaspaandra}
Hemaspaandra, E., Hemaspaandra, L.A.: Dichotomy for voting systems. Journal of
  Computer and System Sciences  73(1),  73--83 (2007)

\bibitem{konczak}
Konczak, K., Lang, J.: Voting procedures with incomplete preferences. In: Proc.
  IJCAI-05 Multidisciplinary Workshop on Advances in Preference Handling.
  vol.~20 (2005)

\bibitem{lu}
Lu, T., Boutilier, C.: Multi-winner social choice with incomplete preferences.
  In: Proceedings of the Twenty-Third international joint conference on
  Artificial Intelligence. pp. 263--270. AAAI Press (2013)

\bibitem{nina}
Narodytska, N., Walsh, T.: The computational impact of partial votes on
  strategic voting. In: Proceedings of the 21st European Conference on
  Artificial Intelligence (ECAI). pp. 657--662 (2014)

\bibitem{procaccia2}
Procaccia, A.D., Rosenschein, J.S.: Junta distributions and the average-case
  complexity of manipulating elections. In: Proceedings of the fifth
  international joint conference on Autonomous agents and multiagent systems.
  pp. 497--504. ACM (2006)

\bibitem{procaccia3}
Procaccia, A.D., Rosenschein, J.S.: Average-case tractability of manipulation
  in voting via the fraction of manipulators. In: AAMAS. p. 105 (2007)

\bibitem{xia}
Xia, L., Conitzer, V.: Determining possible and necessary winners under common
  voting rules given partial orders. Journal of Artificial Intelligence
  Research  41(2),  25--67 (2011)

\end{thebibliography}

\appendix
\section*{Appendix: Proofs}

\subsubsection{Proof of Theorem \ref{FDSS}.} 
 It is easy to see that Fixed-Difference Subset Sum is in NP. To prove NP-hardness, we show a reduction from an arbitrary instance of the Partition (P1) problem. Let the arbitrary instance be $\{k_1, 
\cdots, k_t\}$ with $\sum_i k_i = 2K$. Now, construct the following instance $\{l_1, \cdots, l_t, l'_1, \cdots, l'_t\}$ of Fixed-Difference Subset Sum(P2), where $l_i = k_i + 2^{n + i}$, $l'_i = 
2^{n + i}$, and $n = \lceil \log{2K} \rceil$.

Suppose there exists a partition $S_1, S_2$ for P1. Then P2 has subsets $T_1$, $T_2$ such that $\sum T_1 - \sum T_2 = K + \sum_{i=1}^{t} 2^{i+n}$, where $T_1 = \{l_i\}_{i 
| k_i \in S_1} \cup \{l'_i\}_{i | k_i \in S_2}$ and $T_2 = \emptyset$.

Conversely, suppose there exists subsets $T_1$, $T_2$ in P2 such that $\sum T_1  -  \sum T_2 = K + \sum_{i=1}^{t} 2^{i+n}$. Now, it is easy to argue that none of $l_i$ or $l'_i$ for $i = 1, 
\cdots, t$ can belong to $T_2$, because if so then the second term ($\sum_{i=1}^{t} 2^{i+n}$) of $\sum T_1  -  \sum T_2$ will not be attainable. Therefore, we can construct $S_1$ and $S_2$ such 
that $S_1 = \{k_i\}_{i | l_i \in T_1}$ and $S_2 = \{k_i\}_{i | l_i \notin T_1}$, and this in turn implies that P1 has a partition.
\qed

\subsubsection{Proof of Theorem \ref{rdThm}.} 
Since there are only three candidates, the scoring vector for the corresponding positional scoring rule is defined by $\langle \alpha_1, \alpha_2, \alpha_3 \rangle$, where $\alpha_1 > \alpha_2 > 
\alpha_3 = 0$ (because $\alpha_1 = \alpha_2$ is isomorphic to veto, $\alpha_2 = \alpha_3$ is isomorphic to plurality, and $\alpha_3$ can be taken to be zero since translating the scores in a 
scoring rule does not affect the outcome of the rule). Also, note that if the three candidates are $p, a,$ and $b$, each manipulator votes in one of the following ways\footnotemark : $(p), (p, a, b), 
(p, b, a)$, where for $(p)$ candidate $p$ gets a score $\alpha_2$. 

\footnotetext{$(p, a, b) \equiv (p, a)$, and $(p, b, a) \equiv (p, b)$ }

The problem is in NP since winner determination for any scoring rule can be done in polynomial time. To show NP-hardness, we proceed by considering three cases: 1) $\alpha_1 > \frac{3}{2}\alpha_2$ 
2) $\alpha_1 < \frac{3}{2}\alpha_2$ 3) $\alpha_1 = \frac{3}{2}\alpha_2$. For the first two cases, we reduce an arbitrary instance of the Partition problem to an instance of CWCM, and for the third 
case we show a reduction from the Fixed-Difference Subset Sum problem.  

\begin{case}[{$\alpha_1 > \frac{3}{2}\alpha_2$}] 
Given a Partition instance $\{k_i\}_{1\leq i\leq t}$ summing to $2K$, construct the following instance of CWCM, where $p, a,$ and $b$ 
are the three candidates. In S, let there be a voter of weight $(2\alpha_1 - \alpha_2)K$ voting for $(a, b, p)$ and $(b, a, p)$ each. As a result, $a$ and $b$ have a score of $(2\alpha_1 - 
\alpha_2)(\alpha_1 + \alpha_2)K$ each. In $T$, let each $k_i$ have a vote of weight $(\alpha_1 + \alpha_2)k_i$.

Suppose there exists a partition. Let those manipulators in one partition vote $(p, a, b)$ and those in the other vote $(p, b, a)$. Then the score of $p, a$ and $b$, is $2\alpha_1(\alpha_1 + 
\alpha_2)K$, and so $p$ is a winner. 

Conversely, suppose there exists a manipulation in favor of $p$. Let $x, y$, and $z$ be the sum of the $k_i$'s of the manipulators in $T$ who vote $(p, a, b), (p, b, a),$ and $(p)$, respectively. So 
now, the score of $p$ is $((x  +  y)\alpha_1  +  z\alpha_2)(\alpha_1  +  \alpha_2)$, while that of $a$ and $b$ is $((2\alpha_1  -  \alpha_2)K  +  x\alpha_2)(\alpha_1  +  \alpha_2)$ and $((2\alpha_1  
- \alpha_2)K + y\alpha_2)(\alpha_1  +  \alpha_2)$, respectively. Since there exists a successful manipulation, score of $p$ should be at least as large as that of $a$, and so we have $((x  +  
y)\alpha_1  +  z\alpha_2)(\alpha_1  +  \alpha_2) \geq ((2\alpha_1  -  \alpha_2)K  +  x\alpha_2)(\alpha_1 +  \alpha_2)$. Using the fact that $x + y + z = 2K$, this simplifies to $(K - x)\alpha_2 \geq 
z(\alpha_1 - \alpha_2)$ (1). Again, the score of $p$ should be at least as large as that of $b$, so we have $(K - y)\alpha_2 \geq z(\alpha_1 - \alpha_2)$ (2). Adding (1) and (2) and simplifying it, 
we have $z(2\alpha_1 - 3\alpha_2) \leq 0$. Now, since we assumed  $\alpha_1 > \frac{3}{2}\alpha_2$, this implies that $z \leq 0$. But $z$ cannot be less than 0, so it has to be equal to 0.  Plugging 
$z = 0$ in (1) and (2), we have $x \leq K$ and  $y \leq K$ respectively. This together with the fact that $x + y + z = 2K$  implies that $x = y = K$, and therefore there exists a partition.
\end{case}

\begin{case}[{$\alpha_1 < \frac{3}{2}\alpha_2$}] 
Given a Partition instance, construct the following instance of CWCM. In S, let there be a voter of weight $15K$ voting for $(b, a, p)$, 
a voter of weight $5K$ voting for $(b, p, a)$, a voter of weight $11K$ voting for $(a, p, b)$, a voter of weight $9K$ voting for $(a, b, p)$, and a voter of weight $7K$ voting for $(p, b, a)$ and 
$(p, a, b)$ each. As a result, the scores of $a$, $b$, and $p$ are $(20\alpha_1 + 22\alpha_2)K$, $(20\alpha_1 + 16\alpha_2)K$, and $(14\alpha_1 + 16\alpha_2)K$ respectively. In $T$, let each $k_i$ 
have a vote of weight $6k_i$. 

Suppose there exists a partition. Let those manipulators in one partition (who weight to $6K$) vote $(p, b, a)$ and those in the other vote $(p)$. Then the score of all the three candidates is 
$20\alpha_{1}K + 22\alpha_{2}K$, and so $p$ is a winner. 

Conversely, suppose there exists a manipulation in favor of $p$. Let $x, y$, and $z$ be the total weight of manipulators in $T$ who vote $(p, a, b), (p, b, a),$ and $(p)$ respectively. So now, the 
score of $p$ is $(x  +  y)\alpha_1  +  z\alpha_2 + (14\alpha_{1} + 16\alpha_{2})K$, while that of $a$ and $b$ is $(20\alpha_1 + 22\alpha_2)K + x\alpha_2$ and $(20\alpha_1 + 16\alpha_2)K + y\alpha_2$ 
respectively. Since there exists a successful manipulation, score of $p$ should be at least as large as that of $a$, and so we have $(x  +  y)\alpha_1  +  z\alpha_2 + 14\alpha_{1}K + 16\alpha_{2}K 
\geq (20\alpha_1 + 22\alpha_2)K + x\alpha_2$. Using the fact that $x + y + z = 12K$, this simplifies to $6(\alpha_1 - \alpha_2)K  - x\alpha_2 \geq z(\alpha_1 - \alpha_2)$ (1). Again, the score 
of $p$ should be at least as large as that of $b$, so we have $6\alpha_{1}K - y\alpha_2 \geq z(\alpha_1 - \alpha_2)$ (2). Adding (1) and (2) and simplifying it, we have $(6K - z)(2\alpha_1 - 
3\alpha_2) \geq 0$. Now, since we assumed $\alpha_1 < \frac{3}{2}\alpha_2$, this implies that $(6K - z) \leq 0$, or $z \geq 6K$. Plugging $z \geq 6K$ in (1) and (2), we have $x \leq 0$, and  $y \leq 
6K$, respectively. But then $x$ cannot be less than 0, so it has to be equal to 0, and this in turn results in $z \leq 6K$ in (1). But again, $z$ cannot not be both greater than and lesser than equal 
to $6K$. So, $z$ has to be equal to $6K$, and since  $x + y + z = 12K$, $y = 6K$. This implies there exists a partition.
\end{case}

\begin{case}[{$\alpha_1 = \frac{3}{2}\alpha_2$}] 
Consider the same instance of CWCM as in case 2. The scores of $a$, $b$, and $p$ are $(20\alpha_1 + 22\alpha_2)K$, $(20\alpha_1 + 
16\alpha_2)K$, and $(14\alpha_1 + 16\alpha_2)K$ respectively. In $T$, let each $k_i$ have a vote of weight $6k_i$. 

Suppose there exists $S_1$, $S_2$ such that $\sum S_1 - \sum S_2 = K$. Let those manipulators who are in $S_1$ vote $(p, b, a)$, those in $S_2$ vote $(p, a, b)$, and let all the remaining 
manipulators vote $(p)$. If $x, y,$ and $z$ denote the sum of $k_i$'s of the manipulators who vote for $(p, a, b), (p, b, a),$ and $(p)$, respectively, then the scores of $p, a,$ and 
$b$ are $9(x  + y){\alpha_2}  +  6z\alpha_2 + 37\alpha_{2}K$, ${52\alpha_2}K + 6x\alpha_2$, and $46\alpha_2K + 6y\alpha_2$, respectively. Now if there existed a 
manipulation, then the score of $p$ has to be at least as large as that of $a$ and $b$. Let us consider $p$ and $a$ first. Whatever follows can be replicated for $b$. Suppose $s(p) \geq s(a)$. This 
implies $9(x  + y){\alpha_2}  +  6z\alpha_2 + 37\alpha_{2}K \geq {52\alpha_2}K + 6x\alpha_2$. Simplifying this we have, $x + 3y + 2z \geq 5K$. But since $y - x = K$ and $x + 
y + z = 2K$, we know that $x + 3y + 2z = 5K$, and hence our assumption that $s(p) \geq s(a)$ is true. Doing the same with respect to $p$ and $b$, we will see that $s(p) = s(b)$. As a result, we can 
conclude that existence of $S_1$, $S_2$ such that $\sum S_1 - \sum S_2 = K$ results in a successful manipulation for $p$.

Conversely, suppose there exists a manipulation in favor of $p$. This implies that the score of $p$ is at least as much as that of $a$, and from above we know that this in turn results in the 
inequality $x + 3y + 2z \geq 5K$. Using the fact that $x + y + z = 2K$, we have $y - x \geq K$. Similarly, comparing $p$ and $b$ we have, $y - x \leq K$. But then, $ y - x$ cannot be both greater and 
lesser than equal to $K$ at the same time. So $y - x$ has to be equal to $K$, and this in turn implies that there exists two sets  $S_1, S_2$ such that $\sum S_1 - \sum S_2 = K$, where $y = \sum S_1$ 
and $x = \sum S_2$. \qed
\end{case}

\subsubsection{Proof of Theorem \ref{avThm}.} 
Like in Theorem \ref{rdThm}, since there are only three candidates, the scoring vector for the corresponding positional scoring rule is defined by $\langle \alpha_1, \alpha_2, \alpha_3 
\rangle$, where $\alpha_1 \geq \alpha_2 > \alpha_3 = 0$. Also, note that if the three candidates are $p, a,$ and $b$, each manipulator votes in one of the following ways: $(p), (p, a, b), 
(p, b, a)$, where for $(p)$ candidate $p$ gets a score $\alpha_1$, $a$ and $b$ receive a score of $(\alpha_2/2)$. 

  The proof uses a reduction from the Fixed-Difference Subset Sum problem and is very similar to the one in case 3 Theorem \ref{rdThm}. Given an instance of Partition, construct the following 
instance of CWCM where in $S$ we have a voter of weight $(4\alpha_1 + \alpha_2)K$ voting for $(b, a, p)$, a voter of weight $(2\alpha_1 - \alpha_2)K$ voting for $(a, b, p)$, and a voter of 
weight $2(\alpha_1 + \alpha_2)K$ voting for $(a, p, b)$. As a result, the scores of $a, b,$ and $p$ are $(4\alpha_1 + \alpha_2)(\alpha_1 + \alpha_2)K$, $(4\alpha_1 - \alpha_2)(\alpha_1 + \alpha_2)K$, 
and $2\alpha_2(\alpha_1 + \alpha_2)K$, respectively. In T, let each $k_i$ have a vote of weight $2(\alpha_1 + \alpha_2)k_i$.

Suppose there exists $S_1$, $S_2$ such that $\sum S_1 - \sum S_2 = K$. Let those manipulators who are in $S_1$ vote $(p, b, a)$, those in $S_2$ vote $(p, a, b)$, and let all those remaining vote 
$(p)$. If $x, y,$ and $z$ denote the sum of the $k_i$'s of the manipulators who vote for $(p, a, b), (p, b, a),$ and $(p)$, respectively, then the scores of $p, a,$ and $b$ are $(4\alpha_1 + 
2\alpha_2)(\alpha_1 + \alpha_2)K$, $((4\alpha_1 + \alpha_2)K + 2(x\alpha_2 + z\alpha_2/2) )(\alpha_1 + \alpha_2)$, and $((4\alpha_1 - \alpha_2)K + 2(y\alpha_2 + z\alpha_2/2))(\alpha_1 + \alpha_2)$, 
respectively. Now if there existed a manipulation, then the score of $p$ has to be at least as large as that of $a$ and $b$. Let us consider $p$ and $a$ first. Whatever follows can be replicated for 
$b$. Suppose $s(p) \geq s(a)$. This implies that $(4\alpha_1 + 2\alpha_2)K \geq (4\alpha_1 + \alpha_2)K + 2x\alpha_2 + z\alpha_2$. Simplifying this we have, $2x + z \leq K$. But since $y - x = K$ and 
$x + y + z = 2K$, we know that $2x + z = K$, and hence our assumption that $s(p) \geq s(a)$ is true. Doing the same with respect to $p$ and $b$, we will see that $s(p) = s(b)$. As a result, we can 
conclude that existence of $S_1$, $S_2$ such that $\sum S_1 - \sum S_2 = K$ results in a successful manipulation for $p$. 

Conversely, suppose there exists a manipulation in favor of $p$. This implies that the score of $p$ is at least as much as that of $a$, and from above we know that this in turn results in the 
inequality $2x + z \leq K$ (1). Similarly, comparing $p$ and $b$ we have, $2y + z \leq 3K$ (2). Now, using the fact that $x + y + z = 2K$, we know that inequality (1) reduces to $y - x \geq K$ and 
inequality (2) reduces to $y - x \leq K$. But then, $ y - x$ cannot be both greater and lesser than equal to $K$ at the same time. So $y - x$ has to be equal to $K$, and this in turn implies that 
there exists two sets $S_1, S_2$ such that $\sum S_1 - \sum S_2 = K$, where $y = \sum S_1$ and $x = \sum S_2$.  \qed

\subsubsection{Proof of Theorem \ref{vetoThm}.} 
  Consider an arbitrary set $W$ of top-truncated votes which -- along with the set $S$ of non-manipulators' votes -- results in an elimination order $e = (c_1, c_2, \cdots, c_m = p)$, where $p$ is 
the preferred candidate, and $c_i$ is the candidate eliminated in the $i$th round. Now, consider the set of votes $X$ such that each vote in $W$ is replaced by $p = c_m \succ c_{m-1} \succ \cdots 
\succ c_1$. $X$ along with $S$ results in the same elimination order $e$. Therefore, we see that any manipulation that can be achieved by a set of top-truncated votes can be achieved by casting an 
equivalent set of complete votes. \qed

\subsubsection{Proof of Theorem \ref{antiwcm}.} 
Assume there exists an arbitrary set $W$ of top-truncated votes which results in $d$ receiving the lowest score. For each of the top-truncated votes in $W$, let us complete them in the following way: 
If $d$ is included in the vote, append all the other candidates who are not part of it in any arbitrary order. If $d$ is not there, place $d$ at the bottom of the preference ordering (as the 
$m$th preferred candidate) and the rest in any arbitrary order. Completing the votes as above does not change the candidate with the lowest score. \qed

\subsubsection{Proof of Theorem \ref{copeland-alpha}.} 
 It is easy to show that the problem is in NP. To show that it is NP-hard, we use a reduction from an arbitrary instance of Fixed-Difference Subset Sum problem. Let $p, a,$ and $b$ be the three 
candidates. In $S$, let there be a voter of weight $3K$ voting for $(a, b, p)$ and a voter of weight $K$ voting for $(b, a, p)$. In $T$, let each $k_i$ have a vote of weight $2k_i$.

Suppose there exists $S_1, S_2$ such that $\sum S_1 - \sum S_2 = K$. In Copeland$^{\alpha}$, it can be assumed that all the manipulators rank $p$ first. So, let the manipulators in $S_1$ vote $(p, 
b, a)\footnotemark$, those in $S_2$ vote $(p, a, b)$, and let the rest vote for $(p)$. If $N_{V}(r, s)$ denotes the total number of votes in $V$ which rank $r$ prior to $s$ and $D_{V}(r, s) = 
N_{V}(r, s) - N_{V}(s, r)$, then $D_{S \cup T}(p, a) = 0$ and $D_{S \cup T}(p, b) = 0$. Therefore, the score of $p$, $s(p) = 2\alpha$. Also since $\sum S_1 - \sum S_2 = K$,  $D_{T}(a, b) = -2K$, 
while $D_S(a,b) = 2K$. Therefore, $D_{S \cup T}(a, b) = 0$ and so, both receive a score $2\alpha$. Since all of them have the same score, $p$ is a winner.

\footnotetext{In Copeland$^{\alpha}$ with 3 candidates, $(p, a, b) \equiv (p, a)$, and $(p, b, a) \equiv (p, b)$}

Conversely, suppose there exists a successful manipulation in favor of $p$. If $x, y$, and $z$, denote the sum of $k_i$'s of the manipulators in $T$ who vote $(p, a, b), \allowbreak (p, b, a),$ and 
$(p)$, respectively, then without taking into account the pairwise election between $a$ and $b$ in $T$, the score of $p$, $a$, and $b$ is $2\alpha, 1 + \alpha,$ and $\alpha$, respectively. Now since 
$2\alpha < 1 + \alpha$ for all rational $\alpha \in [0, 1)$, therefore, the only way $p$ would win this is if including the pairwise election between $a$ and $b$ in $T$ results in a tie between them. 
So this implies that $D_{S \cup T}(a, b) = 2K + 2x - 2y = 0$ and that  $y - x = K$. This in turn implies that there exists sets $S_1$ and $S_2$ such that $\sum S_1 - \sum S_2 = K$, where $y = \sum 
S_1$ and $x =  \sum S_2$.\qed

\subsubsection{Proof of Theorem \ref{maxminThm}.} 
 Let $W$ be an arbitrary set of top-truncated votes which -- along with the set S of non-manipulators' votes -- results in $p$ being a winner. Since in the maximin rule moving $p$ to the top will 
never hurt $p$, we can safely assume that all the votes in $W$ have $p$ at the top. Now for every vote in $W$, replace it by $(p)$. By doing so we see that $p$'s score does not change. Also note 
that for all the other candidates their scores can only decrease or stay the same, but can never increase. Therefore, any constructive manipulation achieved for $p$ can be achieved if all 
manipulators just vote $(p)$. \qed

\subsubsection{Proof of Theorem \ref{dwcmThm1}.} 
Consider a voting rule $X$ that is monotone and is based on numerical scores. Let the destructive manipulation be against the candidate $h$. Now, suppose there exists an arbitrary set of 
top-truncated votes $W$ that -- along with the set $S$ of non-manipulators' votes -- results in the destructive manipulation of $h$ in $X$. Since $X$ is based on scores we will have a final ordering 
of the candidates after the election. Let $e: c_1 (\neq h) \succ c_2 \succ \cdots \succ c_m$ denote that ordering. Next, consider the set of votes $W'$ which is formed by completing the votes in $W$ 
in the following way: replace each vote in $W$ by placing $c_1$ at the top, $h$ at the bottom (i.e. at the $m$th position), and the rest of the candidates in any arbitrary order. Since $X$ 
is monotone, $W'$ along with $S$ cannot result in the score of $c_1$ decreasing and nor can it result in the score of $h$ increasing. Therefore, if $W$ resulted in the destructive manipulation of 
$h$ then so should $W'$. \qed 

\subsubsection{Proof of Theorem \ref{eliminateNPCdest}.} 
 Since there are only three candidates, the scoring vector for the corresponding positional scoring rule is defined by $\langle \alpha_1, \alpha_2, \alpha_3 \rangle$, where $\alpha_1 > \alpha_2 \geq  
\alpha_3 = 0$. Showing that the problem is in NP is easy. To show NP-hardness, we use a reduction from an arbitrary instance of Anti-WCM$\langle S, T, h \rangle$ with 3 candidates (see Corollary 
\ref{corollaryAntiWCM}), where $a, b$, and $h$ are the three candidates, and $h$ is the disliked candidate. In the DWCM instance we construct, we use the same set of candidates, the same set of 
manipulators $T$, and to the $S$ from the Anti-WCM instance we add the following set $S'$ of voters such that $K$ is greater than the sum of the weights in $S$ and $T$ combined. In each case, we add 
1 voter of the corresponding type and weight specified below. 
\begin{align*}
 K&: (a, h, b) \\ 2K&: (h, a, b) \\ K&: (b, h, a) \\ 2K&: (h, b, a) \\ 3K&: (a) \\ 3K&: (b)
\end{align*}
We set $h$ to be the disliked candidate in the DWCM instance.

Suppose there was a way to make $h$ receive the lowest score in the Anti-WCM instance. If score$_{S}(a)$ denotes the score candidate $a$ receives from $S$, then this implies that score$_{S \cup 
T}(a)$ $>$ score$_{S \cup T}(h)$, and score$_{S \cup T}(b)$ $>$ score$_{S \cup T}(h)$. Also, note that all the three candidates $a, b$, and $h$ are tied in $S'$ as each of them receive a score of 
$4\alpha_1K + 2\alpha_2K$. This in turn implies that $h$ receives the lowest score in the DWCM instance, and so will be eliminated in the first round. Thus, existence of a successful manipulation in 
the Anti-WCM instance ensures the existence of a successful manipulation in the DWCM instance. 

Conversely, suppose there exists a successful destructive manipulation against $h$ in the DWCM instance. We first show that this is possible only if $h$ is eliminated in the first round in 
eliminate$(X)$. To do so, let us assume it were not the case and that one of $a$ or $b$ was eliminated in the first round. Let us consider $a$ first. If $a$ was eliminated in the first round then the 
votes in $S'$ would now be: 
\begin{align*}
 K&: (h, b) \\ 2K&: (h, b) \\ K&: (b, h) \\ 2K&: (h, b) \\ 3K&: (b)
\end{align*}

Since the elimination of $a$ means that they are only two candidates remaining, from now on we can assume our protocol to be equivalent to plurality (since any scoring rule is equivalent to plurality 
when there are only two candidates). So now, score$_{S'}(h)$ $-$ score$_{S'}(b) = K$ and since $K$ is greater than sum of the weights in $S$ and $T$ combined this implies that in the subsequent round 
$b$ will be eliminated, thus resulting in $h$ winning the DWCM instance. Therefore, there cannot be a destructive manipulation against $h$ in the DWCM instance if $a$ is eliminated in the first 
round. Similarly, by doing things identically for candidate $b$, we can see that a destructive manipulation against $h$ will not be possible if $b$ is eliminated in the first round. This in 
turn leads us to conclude that a successful destructive manipulation against $h$ is possible only if $h$ is eliminated in the first round. But then, since all the three candidates are tied in $S'$, 
the only way this can happen is if $h$ receives the lowest score in $S$. Or in other words, a successful destructive manipulation against $h$ in the DWCM instance is possible only if there 
exists a successful manipulation against $h$ in the Anti-WCM instance.  \qed

\subsubsection{Proof of Theorem \ref{extbribEval}.} 
 Consider an arbitrary instance of the Weighted Extension-bribery with $S$, and $p$, and the same instance for the Evaluation problem. Now, it is clear that if there exists an extension in Weighted 
Extension-bribery problem, then $p$ wins with probability greater than zero in the Evaluation problem. Conversely, if $p$ wins the non-zero probability in the Evaluation problem then this implies 
that there is at least on extension where it wins. \qed

\subsubsection{Proof of Theorem \ref{ext-eval}.} 
 For all the above mentioned protocols except plurality$_{\downarrow}$, the best we can do is to extend each top-truncated vote by placing $p$ at its end (i.e. place $p$ as $k$th candidate if $k-1$ 
candidates are already ranked), if it isn't already present. In case of plurality$_{\downarrow}$, the best strategy is to complete each of the top-truncated votes by placing $p$ at the topmost 
position possible followed by all the other as-yet unranked candidates in any arbitrary order. \qed

\subsubsection{Proof of Theorem \ref{elVeto}.} 
Suppose there was an arbitrary extension $W$ which resulted in $p$ winning. Let the corresponding elimination order be $e = (c_1, c_2, \allowbreak \cdots, c_m = p)$, where $c_i$ is the candidate 
eliminated in the $i$th round. Now, we claim that the same elimination order can be achieved by doing the following:
\begin{itemize}
 \item In each of the of the top-truncated vote, complete it by placing the candidates in the reverse order in which they appear in $e$. That is, place $p$ if not already present, followed by 
$c_{m-1}$ if not present, and so on until $c_1$.
\end{itemize}
Doing the above results in the same elimination order as $e$. This can be shown through an inductive argument. When there are $m$ candidates, $c_1$ which is eliminated first in $e$ has been 
placed last wherever possible and this in turn will result it in getting eliminated first in our completion. Once $c_1$ is eliminated, we have $m-1$ candidates with $c_2$ placed last wherever 
possible and therefore $c_2$ will be eliminated next. Continuing this way, guarantees the elimination order $e$.

To solve Weighted Extension-bribery, the campaign manager can try out all possible elimination orders, extend the votes as outlined above, and see if any of them results in $p$ winning. Doing so will 
only take polynomial time since the number of candidates are bounded.  \qed

\subsubsection{Proof of Theorem \ref{eVCWIM}.} 
 Since the Evaluation problem is in $P$, the manipulator can try out all possible complete orders to check if manipulation is possible. This is enough because, for eliminate(veto), any 
manipulation that can be achieved by top-truncated votes can be achieved by completing the vote appropriately (see Theorem \ref{vetoThm}). \qed

\subsubsection{Proof of Theorem \ref{eV-WE}.} 
 We show this by a reduction from an arbitrary instance of the Partition problem to the following instance of Weighted Evaluation. Let $a, b,$ and $p$ be three candidates. Let there be a vote of 
weight 1 for $(p, a, b)$. For each $k_i$ in the partition instance, let it have a weight $k_i$ and vote for $(a, p, b)$ and $(b, p, a)$ with probability 1/2 each.

Now, we can see that $p$ wins if and only if $(a, p, b)$ and $(b, p, a)$ are voted by exactly $K$ of the vote weight, because failing to do so would mean that $p$ will be eliminated in the second 
round. But then this is possible if and only if there exists a partition. \qed

\end{document}